\documentclass[journal]{IEEEtran}

\usepackage{cite}
\usepackage{amsmath,amssymb,amsfonts}
\usepackage{algorithmic}
\usepackage{graphicx}
\usepackage{algorithm,algorithmic}
\usepackage{hyperref}
\hypersetup{hidelinks=true}
\usepackage{textcomp}

\usepackage{amsthm}
\newtheorem{definition}{Definition}
\newtheorem{theorem}{Theorem}
\newtheorem{lemma}{Lemma}

\def\BibTeX{{\rm B\kern-.05em{\sc i\kern-.025em b}\kern-.08em
    T\kern-.1667em\lower.7ex\hbox{E}\kern-.125emX}}
\title{Bounds on Prediction Error When Using an Impulse Response/Equilibrium Model Structure}
\author{Tyrone L. Vincent and Michael B. Wakin%
\thanks{T. L. Vincent and M. B. Wakin are with the Department of Electrical Engineering, Colorado School of Mines, Golden, CO 80401.}
\thanks{This research was supported by the Office of Naval Research via grants N00014-23-1-2694 and N00014-25-1-2078 and the National Science Foundation via grant CCF-2106834.}}
\date{\today}

\newcommand{\meas}[1]{\widetilde{#1}}
\newcommand{\est}[1]{\widehat{#1}}
\begin{document}
\maketitle

\begin{abstract}
    An impulse response/equilibrium model (IREM) structure combines a linear convolution model with a nonlinear function that sets the current operating point via an equilibrium variable with integrator dynamics. This model structure is well suited for mildly nonlinear systems and in particular has been applied to battery fast charging control. This paper provides observability conditions for the IREM model structure and bounds on the prediction error. These conditions can be evaluated directly on the system impulse response.
\end{abstract}

\section{Introduction}

This paper concerns estimation and prediction for dynamical systems. These are key tools for model predictive control (MPC), where control actions are determined via optimization. In order to form this optimization problem, it is necessary to predict the future behavior of a given system in response to a proposed future input. This problem is complicated by two factors: First, the behavior of a dynamical system depends on past inputs, including unmeasured disturbances. Second, the available sensors may be restricted so that not all states and/or signals of interest are measured. Both of these problems are solved by a process of state estimation, where the system state is defined to be the set of variables that, if known at the current time, are sufficient to predict the future outputs given future inputs. A classical method for state estimation uses a state-space representation of the system, and the current state is estimated based on measured past inputs and outputs using one of a large number of state estimation methods under the name of observer theory and/or Kalman Filtering~\cite{kailath2000linear,Boutat2021Observer,Gauthier2009Deterministic,brown2012introduction}. However, one requirement for implementation of these classical methods is the availability of a (preferably low order) state-space system representation, which may not be immediately available.  

The goal of this paper is to study prediction methods that can be implemented for mildly nonlinear systems using a small amount of data obtained from a system simulation, and without requiring a model identification or reduction step.  Of course, there are many data-based prediction methods available, including neural networks~\cite{PARLOS2000765,mohajerin2019multistep}, Koopman operator methods~\cite{Korda2020}, and other machine learning approaches, but although they can be applied to a large class of nonlinear systems, in general they require a large amount of data and significant processing of that data. We examine here an alternative that is appropriate for a smaller class of systems but can be implemented using a very limited set of experiments, with no additional processing of the data needed.

In this paper, we make use of a mixed input/output and state representation, where the states include only integrator dynamics. Specifically, the model structure is 
\begin{equation}
\begin{aligned}\label{eqn:irem}
y_{k} & = \sum_{\ell=-\infty}^{k}g_{k-\ell}u_{\ell} + \phi(x_{k})\\
x_{k} &= x_{k-1}+f(u_{k})
\end{aligned}
\end{equation}
where $u_{k}\in\mathbb{R}^{m}$ is the system input, $y_{k}\in\mathbb{R}^{q}$ is the system output, $g_{k}\in\mathbb{R}^{q\times m}$ is a sequence characterizing the system dynamics around an equilibrium (essentially the system impulse response modulo the integrator dynamics), $x_{k} \in\mathbb{R}^{s}$ is a state with integrator dynamics that tracks the system equilibrium, $f:\mathbb{R}^{m}\rightarrow \mathbb{R}^{s}$ maps the input to the equilibrium state, and $\phi:\mathbb{R}^{s}\rightarrow \mathbb{R}^{q}$ is a function that maps the equilibrium state to the system output at that equilibrium. This structure is called an impulse response/equilibrium model (IREM). 
For the problem of prediction, the sequence $g_{k}$ and  functions $\phi$ and $f$ are assumed to be known.

As a specific example of utilizing this model structure, consider a battery pack, where the input is the applied current in amps, and the outputs are the terminal voltage, along with other important variables relevant for ensuring the battery is operated without damage, such as spatial ion concentration or internal voltages (i.e., anode to electrolyte). Typically only applied current and terminal voltage can be measured. With zero current, the battery outputs relax to an equilibrium value determined by the state of charge. In the IREM model structure, we choose $x_{k}$ to be the state of charge, and $f(u_{k}) = \Gamma u_{k}$ where $\Gamma$ is the scale factor $T_{s}/(3600Q)$, $Q$ is the battery capacity in Ah, and $T_{s}$ is the sample time in seconds. The function $\phi$ can be found via simulation, either by finding the equilibrium at several different states of charge and interpolating, or by simulating a continuous charge or discharge at a very small current and recording the outputs vs.\ state of charge. Finally $g_{k}$ can be found from the simulation of a pulse of duration $T_{s}$ and unit area or a unit step response. In either case, first $\phi(x_{k})$ is subtracted from the response. If the input was a pulse, $g_{k}$ is the resulting output sequence, while if the input was a step, $g_{k}$ is the first difference of the output sequence. In practice, one would do this simulation at several different states of charge, and then schedule $g_{k}$ on the value of $x_{k}$, however our analysis will be for the unscheduled model structure given in \eqref{eqn:irem}.  Kalman filtering has been a popular approach to the estimation of battery state of charge and other internal variables \cite{bib:P2004,bib:SW2006,bib:SATP2015}, but it requires significant effort in model reduction to obtain a suitable state-space model \cite{bib:LASP2014}. Battery estimation of both state of charge and other internal variables using IREM, with the model parameters obtained via simple step experiments, was demonstrated in \cite{evans2025}.

In what follows,  we discuss the process of performing prediction using the IREM model structure and establish bounds on the prediction error. Importantly, we do not assume that all outputs are measured, but that only a subset are available. While the results do not assume a particular method of choosing the estimate, they are most compatible with a moving horizon approach, whereby a window of data in the past is used to determine the current state and then predict output trajectories in the future (both measured and unmeasured). The main result provides a gain that bounds the prediction error over a window in the future in terms of the fit to measured data over the past window. This gain depends only on the model parameters, and it does not depend on a probabilistic model of noise and disturbances. Moving horizon estimation also provides a more natural approach to gain scheduling, since a different model can be used at each step.  

There is extensive literature on the stability of moving horizon estimation for general nonlinear systems, with some representative results including \cite{rao2003constrained,ji2015robust,muller2017nonlinear,schiller2023lyapunov}. These results establish the convergence of estimates under certain assumptions about the noise and disturbances as well as observability or detectability conditions on the nonlinear system. In this paper, we provide bounds on performance along with observability conditions that are tailored to the specific IREM model structure and are more easily verified. Because the bounds are expressed in terms of the relationship between past fit and future prediction, they are of a different form than those available in the existing literature. 

\section{Notation}

$|\cdot|$ is absolute value, and $\|\cdot\|$ denotes the Euclidean norm. Matrices and vectors may both be denoted by a lower case or capital letter, and will be defined in context. Given matrix $M$, let 
\[
\sigma_{\rm max}(M) = \max_{\|x\|=1}\|Mx\|
\]
be the maximum gain when mapping a unit vector. This is equal to the maximum singular value of $M$. Let
\[
\sigma_{\rm min}(M) = \min_{\|x\|=1}\|Mx\|
\]
be the minimum gain when mapping a unit vector. When $M$ is full column rank, this is given by the smallest singular value. If $M$ is not full column rank, this is zero. Given vector $x$, $\operatorname{diag}x$ is a diagonal matrix with the elements of $x$ along the diagonal. Let $\lambda \searrow 0$ denote convergence to zero from above. $A^{+}$ denotes the pseudo-inverse of matrix $A$.

\begin{lemma}
    \label{lem:pseduoinverse}
    The pseudo-inverse of matrix $A$ is given by
    \[
    A^{+} = \lim_{\lambda\searrow 0}(A^{T}A + \lambda I)^{-1}A^{T} = \lim_{\lambda \searrow 0} A^{T}(AA^{T}+\lambda I)^{-1}.
    \]
    
\end{lemma} Note that these limits exist even if $A$ is not full rank~\cite{golub2013matrix}. 
\begin{definition} Given matrices $A\in\mathbb{R}^{n\times n}$ and $C\in\mathbb{R}^{q\times n}$, $(A,C)$ is an observable pair if the matrix
\[
\mathcal{O} = \begin{bmatrix}C^{T} & (CA)^{T} & \cdots & (CA^{n-1})^{T}\end{bmatrix}^{T}
\]
is full column rank.
\end{definition}

\section{Linear Case}

Since estimation using impulse response models is not common, we discuss the linear part of the model in isolation first. As is well known, the input-output behavior of all discrete-time linear time-invariant systems can be represented as a convolution of the input with the impulse response of the system. Thus, the relationship between the input $u_{k}$ and the output $y_{k}$ is given by
\begin{equation}
\label{eqn:impulse}
y_{k}  = \sum_{\ell=-\infty}^{k}g_{k-\ell}u_{\ell},
\end{equation}
where $g_{k}$ is the impulse response. We will restrict ourselves to causal systems, so $g_{k}=0$ for $k<0$. Prediction using a convolution representation is utilized in the MPC method named Dynamic Matrix Control, however, it does not do so via a method equivalent to state estimation \cite{garcia1986quadratic,garcia1989model}. Instead, the prediction uses a finite length approximation to $g_k$, applies past inputs, and compares to past outputs to estimate a constant disturbance. This method requires that all outputs be measured, and of course, the system dynamics are approximated. 

Developing an estimation method that is not approximate and does not rely on knowing the infinite past requires the use of an appropriate notion of state. This is provided by Nerode equivalence, which is developed for linear time-invariant systems in \cite{kailath1980linear}. In this formulation, inputs over the past time that produce the same future output (with zero future input) are assigned to the same equivalence class, which is identified as a system state. The dimension of the equivalence class then becomes the dimension of the system. For linear systems in a convolution representation, the  system dimension also corresponds to a rank condition on a Toeplitz (or Hankel) matrix made up of the impulse response. Define the Toeplitz matrix
\[
 T_{j}^{\ell,m}(\left\{g_{k}\right\}) = \begin{bmatrix} g_{j+\ell-1} & g_{j+\ell} & \cdots & g_{j+\ell+m-2} \\
g_{j+\ell-2} & g_{j+\ell-1} &\cdots & g_{j+\ell+m-3} \\
\vdots  & \vdots & & \vdots \\
g_{j} & g_{j+1} & \cdots & g_{j+m-1}\end{bmatrix}.
\]
If \eqref{eqn:impulse} is $n$ dimensional, then for any integers $j\geq1$, $\ell\geq 1$ and $m\geq 1$,  ${\rm rank}\, T_{j}^{\ell,m}\leq n$. In \cite{vincent2023prediction} this fact was leveraged to develop a prediction method for systems with an impulse-response representation  based on a finite length equivalent input. That is, for any $r\geq n$ there exists a length $r$ sequence $u^{\rm e}_{k}$ such that  \eqref{eqn:impulse} is equivalent to
\begin{equation}\label{eqn:impulseequivalent}
y_{k}  = \sum_{\ell=-p}^{k}g_{k-\ell}u_{\ell} + \sum_{\ell=-p-r}^{-p-1}g_{k-\ell}u^{\rm e}_{\ell}
\end{equation}
for $k\geq -p$. The infinite length sequence $u_{k}$ for $k\leq -p$ is replaced by the equivalent finite length sequence $u^{\rm e}_{k}$ which is a member of the same equivalence class and drives the system to the same state at time $k=-p$. 

Let $\meas{y}\in\mathbb{R}^{q_{m}}$ be the subset of $y$ that is measured, and $\meas{g}_{k}$ be the rows of $g_{k}$ corresponding to $\meas{y}$. In order to perform prediction, the following process can be used. Without loss of generality, let the current time index be $k=0$.  Using recorded input/output data from $k=-p$ to $k=0$, solve for $u_k^{\rm e}$. Then use \eqref{eqn:impulseequivalent} with this $u_k^{\rm e}$ to find $y_{k}$ for $k>0$.

In order to present this process formally, matrix-vector notation will be used. Define the upper triangular Toeplitz matrix
\[
T_{p}(\{g_{k}\}) = \begin{bmatrix} g_{0} & g_{1} & \cdots & g_{p-1}\\
0 & g_{0} & \cdots & g_{p-2}\\
\vdots & \vdots & \ddots & \vdots \\
0& 0 & \cdots & g_{0} \end{bmatrix},
\]
and the vectors
\begin{align*}
U_{\rm p} & =  \begin{bmatrix}  u_{0}^{T} &  u_{-1}^{T} & \cdots &  u_{-p+1}^{T}\end{bmatrix}^{T}, \\
\meas{Y}_{\rm p}&=\begin{bmatrix}   \meas{y}_{0}^{T} & \meas{y}_{-1}^{T} & \cdots & \meas{y}_{-p+1}^{T}\end{bmatrix}^{T}, \\
U_{\rm f} & =  \begin{bmatrix}  u_{p_{f}}^{T} &  u_{p_{f}-1}^{T} & \cdots &  u_{1}^{T}\end{bmatrix}^{T}, \\
Y_{\rm f}&=\begin{bmatrix}   y_{p_{f}}^{T} & y_{p_{f}-1}^{T} & \cdots & y_{1}^{T}\end{bmatrix}^{T}.
\end{align*}
The subscript p denotes past, recorded data, while the subscript f denotes data in the future. The window size in the past is given by $p$ while the window size in the future is $p_f$. 

The following result establishes that the estimation process described above is valid. To state the required conditions on the system, we first relate them to observability of a state-space representation of the system, although later this will be expressed entirely in terms of the impulse response. When $g_{k}$ is the impulse response of an $n$-dimensional system, there exists an $n$-dimensional state-space representation, which can be obtained from the impulse response via the Ho-Kalman algorithm \cite{ho1966effective}. Let the resulting state-space representation be
\begin{equation}\label{eqn:ss}
\begin{aligned}
x_{k+1} &= Ax_{k}+ Bu_{k},\\
\begin{bmatrix}\meas{y}_{k}\\\bar{y}_{k}\end{bmatrix} &= \begin{bmatrix}\meas{C}\\\bar{C}\end{bmatrix}x_{k}+\begin{bmatrix}\meas{D}\\\bar{D}\end{bmatrix}u_{k},
\end{aligned}
\end{equation}
and note that the impulse response $\meas{g}_{k} = \meas{C}A^{k-1}B$ for $k\geq 1$ and $g_{k} = CA^{k-1}B$ for $k\geq 1$ where $C = \begin{bmatrix}\meas{C}^{T} & \bar{C}^{T}\end{bmatrix}^{T}$. 

\begin{theorem} \cite{vincent2023prediction}\label{thm:result2} Given a system \eqref{eqn:impulse} with $g_{k}$ $n$ dimensional, so that there exists minimal state-space representation \eqref{eqn:ss}, suppose $(A,\meas{C})$ is an observable pair, and $p\geq n$. Let $\meas{Y}_{\rm p}$, $Y_{\rm f}$ be generated by \eqref{eqn:impulse} with corresponding $U_{\rm p}$, $U_{\rm f}$.  Choose  $r\geq n$. Then there exists $ U^{\rm e}\in\mathbb{R}^{mr}$ that satisfies
\begin{equation}\label{eqn:past}
 \meas{Y}_{\rm p} = T_{p}(\{\meas{g}_{k}\}) U_{\rm p} + T_{1}^{p,r}(\{\meas{g}_{k}\}) U^{\rm e}.
\end{equation}
In addition, any $ U^{\rm e}$ that satisfies \eqref{eqn:past} will also satisfy
\begin{equation}\label{eqn:future}
Y_{\rm f} = T_{p_{f}}(\{g_{k}\}) U_{\rm f} + T_{1}^{p_{f},p}(\{g_{k}\}) U_{\rm p}
+T_{p+1}^{p_{f},r}(\{g_{k}\}) U^{\rm e}.
\end{equation}
\end{theorem}

Note that \eqref{eqn:past} and \eqref{eqn:future} are both equivalent to \eqref{eqn:impulseequivalent} for different ranges of $k$, where convolution is replaced by multiplication by a Toeplitz matrix.  Thus, the estimation process can be restated as follows: use \eqref{eqn:past} to solve for $U^{\rm e}$, then use the resulting $U^{\rm e}$ in \eqref{eqn:future} to predict the future. In \cite{vincent2023prediction}, an estimation process that models the noise and disturbances as Gaussian sequences was used to derive a maximum likelihood estimate for $U^{\rm e}$. In this paper, we derive bounds on estimation error for estimates obtained by any method.

There are two main issues to be addressed.

\begin{itemize}
\item Only a subset of the outputs of interest are included in  \eqref{eqn:past}, yet the prediction \eqref{eqn:future} is of all outputs, thus the requirement that $(A,\meas{C})$ is an observable pair. Is there a condition that can be expressed in terms of the impulse response alone that will ensure that the behavior of all outputs can be inferred from a given subset of outputs?
\item An exact solution to  \eqref{eqn:past} may not be obtained. What effect will this have on the predicted outputs? 
\end{itemize}
One way an exact solution to \eqref{eqn:past} may not occur is if $mr<n$, so that the dimension of $U^{\rm e}$ is less than the dimension of the equivalence class. This might be desired as a method of model reduction if $n$ is very large. Alternatively, the measured past outputs may be corrupted by measurement noise, which will also imply that an exact solution to  \eqref{eqn:past} is not obtained, as the left hand side will be  $\est{Y}_{p}+E$ where $E$ is due to the difference between the measured output and actual output.  

Suppose an estimate of $U^{\rm e}$ is obtained, denoted $\est{U}^{\rm e}$. The following result bounds the error in using this estimate to predict $Y_{\rm f}$ in terms of how well it fits the past data $\meas{Y}_{\rm p}$, and a gain that depends on the alignment of the null spaces of $T_{1}^{p,r}(\{\meas{g}_{k}\})$ and $T_{p+1}^{p_{f},r}(\{g_{k}\})$. 
Given the estimated equivalent input, $\est{U}^{\rm e}$, let the estimate of past measured outputs,  $\est{\meas{Y}}_{\rm p}$ and predictions of future $\est{Y}_{\rm f}$ be given by
\begin{equation}\label{eqn:estimate}
 \est{\meas{Y}}_{\rm p} = T_{p}(\{\meas{g}_{k}\}) U_{\rm p} + T_{1}^{p,r}(\{\meas{g}_{k}\}) \est{U}^{\rm e}
\end{equation}
and
\begin{equation}\label{eqn:predict}
\est{Y}_{\rm f} = T_{p_{f}}(\{g_{k}\}) U_{\rm f} + T_{1}^{p_{f},p}(\{g_{k}\}) U_{\rm p}
+T_{p+1}^{p_{f},r}(\{g_{k}\}) \est{U}^{\rm e}.
\end{equation}
In order to reduce the notational burden of what follows, define (with some abuse of notation) $\meas{T}_{1}^{p,r}:=T_{1}^{p,r}(\{\meas{g}_{k}\})$ and $T_{p+1}^{p_{f},r}:=T_{p+1}^{p_{f},r}(\{g_{k}\})$, and similarly for $T_{p},\meas{T}_{p}$. That is, we drop the argument, but denote variables created with $\meas{g}_{k}$ with a tilde, and those created with $g_{k}$ without a tilde.
\begin{theorem}\label{thm:maintheoremlinear}
Given \eqref{eqn:past}, \eqref{eqn:future}, \eqref{eqn:estimate}, and \eqref{eqn:predict}, with $U_{\rm p}$ and $U_{\rm f}$ fixed. If ${\rm Null}(\meas{T}_{1}^{p,r})\subseteq{\rm Null}\left(T_{p+1}^{p_{f},r}\right)$, then
\[
\|\est{Y}_{\rm f} - Y_{\rm f}\| \leq \sigma_{\rm max}\,(T_{p+1}^{p_{f},r}(\meas{T}_{1}^{p,r})^{+}) \|\est{\meas{Y}}_{\rm p} - \meas{Y}_{\rm p}\|.
\]
\end{theorem}

The proof of this theorem is presented in Section \ref{sec:proofs}. This result can be used to provide answers to the questions above. The answer to the first question is given by the condition ${\rm Null}(\meas{T}_{1}^{p,r})\subseteq {\rm Null}(T_{p+1}^{p_{f},r})$. This can be verified in a variety of ways, including checking if $\lim_{\lambda\searrow 0}\meas{T}_{1}^{p,r}((T_{p+1}^{p_{f},r})^{T}T_{p+1}^{p_{f},r} + \lambda I)^{-1}$ is bounded. The answer to the second question is that prediction errors are bounded by a term that scales with the size of the error in the fit to past data, with the gain given by $\sigma_{\rm max}\,(T_{p+1}^{p_{f},r}(\meas{T}_{1}^{p,r})^{+})$. In particular, if $\est{\meas{Y}}_{\rm p} = \meas{Y}_{\rm p} + E$ because the past data was measured with noise, then $\|\est{Y}_{\rm f} - Y_{\rm f}\| \leq \sigma_{\rm max}\,(T_{p+1}^{p_{f},r}(\meas{T}_{1}^{p,r})^{+})\|E\|$.

The condition ${\rm Null}(\meas{T}_{1}^{p,r})\subseteq {\rm Null}(T_{p+1}^{p_{f},r})$ can be connected to the observability condition on matrices $(A,\meas{C}$) from the equivalent state-space representation. Thus, although Theorem~\ref{thm:result2} is stated using an observability condition based on an equivalent state-space representation, this condition can be verified directly from the impulse response.

\begin{theorem}\label{thm:ssobservability}
    Suppose a system with impulse response $g_{k}$ has state-space representation \eqref{eqn:ss}. Then there exists $p$ such that ${\rm Null}(\meas{T}_{1}^{p,r})\subseteq {\rm Null}(T_{p+1}^{p_{f},r})$ if $(A,\meas{C})$ is an observable pair.
 \end{theorem}
Proof:
By assumption, the impulse response coefficients are given by $g_{k} = CA^{k-1}B$ and $\meas{g}_{k} = \meas{C}A^{k-1}B$. Let $\mathcal{C} = \begin{bmatrix} B & AB & A^2B &\cdots & A^{r-1}B \end{bmatrix}$ and define 
\[ 
\mathcal{O}^{p}(A,C) := \begin{bmatrix} (CA^{p-1})^{T} & (CA^{p-2})^{T} & \cdots & C^{T} \end{bmatrix}^{T}.
\]
We can write
\begin{align*}
\meas{T}_{1}^{p,r} &= \mathcal{O}^{p}(A,\meas{C})\mathcal{C}, \\
T_{p+1}^{p_{f},r}& =\mathcal{O}^{p_{f}}(A,C)A^{p}\mathcal{C}.
\end{align*}
If $(A,\meas{C})$ is an observable pair, for $p$  sufficiently large, $\mathcal{O}^{p}(A,\meas{C})$ is full column rank.
Thus  ${\rm Null}(\meas{T}_{1}^{p,r}) = {\rm Null}(\mathcal{C})$ while clearly ${\rm Null}\left(T_{p+1}^{p_{f},r}\right) \supseteq {\rm Null}(\mathcal{C})$, proving the result.\hfill \qed\\

\section{Nonlinear Case}

Now we can consider the IREM model structure \eqref{eqn:irem}. Let $\meas{\phi}(x)$ be the elements of $\phi(x)$ corresponding to $\meas{y}$. Introduce the notation
\[
\Phi_{\rm p}(x,U_{\rm p}) = \begin{bmatrix} \meas{\phi}(x) \\ \meas{\phi}(x - f( u_{0})) \\ \vdots \\ \meas{\phi}\left(x - \sum_{i=-p+2}^{0}f(u_{i})\right)\end{bmatrix},
\]
and
\[
\Phi_{\rm f}(x,U_{\rm f})  = \begin{bmatrix} \phi\left(x +\sum_{i=1}^{p_{f}}f(u_{i})\right) \\ \phi\left(x + \sum_{i=1}^{p_{f}-1}f(u_{i})\right) \\ \vdots \\ \phi\left(x +  f(u_{1})\right)\end{bmatrix}.
\]
The past measured outputs are given by
\begin{equation}\label{eqn:altpast}
 \meas{Y}_{\rm p} = \meas{T}_{p} U_{\rm p} + \meas{T}_{1}^{p,r} U^{\rm e} + \Phi_{\rm p}(x_{0},U_{\rm p}),
\end{equation}
while future outputs are given by
\begin{equation}\label{eqn:altfuture}
Y_{\rm f} = T_{p_{f}} U_{\rm f} + T_{1}^{p_{f},p} U_{\rm p}
+T_{p+1}^{p_{f},r} U^{\rm e} + \Phi_{\rm f}(x_{0},U_{\rm f}).
\end{equation}
Suppose there exist estimates $\est{U}^{\rm e}$ and $\est{x}_{0}$, with predicted past and future outputs
\begin{align}\label{eqn:altestimate}
 \est{\meas{Y}}_{\rm p} &= \meas{T}_{p} U_{\rm p} + \meas{T}_{1}^{p,r} \est{U}^{\rm e} + \Phi_{\rm p}(\est{x}_{0},U_{\rm p}) \\
 \label{eqn:altpredict}
\est{Y}_{\rm f} &= T_{p_{f}} U_{\rm f} + T_{1}^{p_{f},p} U_{\rm p}
+T_{p+1}^{p_{f},r} \est{U}^{\rm e} + \Phi_{\rm f}(\est{x}_{0},U_{\rm f}). 
\end{align}

Results are obtained based on two different assumptions on the nonlinearity $\phi$. In one case, we make a strong assumption that $\phi$ shifted to the origin at all points in its domain is sector bounded. A second result applies to a larger class of nonlinear functions, called $L$-smooth, but imposes additional requirements on the input sequence $u_{k}$, which is expressed in terms of the trajectory of the equilibrium state $x_{k}$.

\begin{definition}    $\phi:\mathbb{R}^{s}\rightarrow \mathbb{R}^{q}$ is a {\em sector bounded function with $\gamma$-approximation $B \in \mathbb{R}^{q\times s}$} if there exists a matrix-valued function $R(x,x_{0}): \mathbb{R}^{s} \times \mathbb{R}^{s} \rightarrow \mathbb{R}^{q \times s}$ with $\| R(x,x_{0}) \| \leq \gamma$ such that $\phi(x)-\phi(x_{0})= B(x-x_{0})+R(x,x_{0})(x-x_{0})$ for all $x,x_0 \in \mathbb{R}^s$. 
\end{definition}

\begin{definition} The vector valued function 
$\phi:\mathbb{R}^{s}\rightarrow \mathbb{R}^{q}$ is said to be {\em $L$-smooth} if its Jacobian matrix is uniformly Lipschitz with constant $L$, i.e., if 
\begin{equation}
\| J(x) - J(x_0) \| \le L \| x-x_0 \|
\label{eq:Lsmooth}
\end{equation}
for all $x,x_0 \in \mathbb{R}^{s}$. Here, $J(x)$ is the $q \times s$ Jacobian of $\phi$ evaluated at $x$, and $\| J(x) - J(x_0) \|$ denotes the operator norm of the difference $J(x) - J(x_0)$.
\end{definition}

A periodic band-limited signal can be expressed as a sum of complex exponentials with exponents $j\omega_{i} k$ with $|\omega_{i}|<\omega_{0}$ for some $\omega_{0}$. This can also be expressed as a sum of exponentials $(e^{j\omega_{i}})^{k}$ with $|e^{j\omega_{i}}-1|<|e^{j\omega_{0}} -1|$. In this paper, band-limited signals are defined similarly, but with a representation by more general basis terms.  

\begin{definition}
    A signal $x_{k}$ is said to be {\em band-limited} if it can be expressed as 
    \[
    x_{k} = \sum_{i=1}^{m}\sum_{\ell=0}^{\ell_{i}}K_{i,\ell}k^{\ell}\lambda_{i}^k
    \]
    with $\lambda_{i}\in\mathbb{C}$,  $|\lambda_{i}-1|\leq {\rho}$, $\forall i$.
\end{definition}

The final ingredient necessary to state the main result is a particular factorization of the nonlinear terms. 

\begin{lemma}
\label{lem:MPsi}
Consider the function $\meas{\phi}:\mathbb{R}^{s}\rightarrow \mathbb{R}^{q_{m}}$ with $U_{\rm p}$ and $U_{\rm f}$ fixed. Let $M(x,x_{0}): \mathbb{R}^{s} \times \mathbb{R}^{s} \rightarrow \mathbb{R}^{q_{m} \times s}$ be such that
\[
\meas{\phi}(x)-\meas{\phi}(x_{0}) = M(x,x_{0})(x-x_{0}),
\]
with one such $M(x,x_{0})$ given by
\[
M(x,x_{0}) = \frac{(\meas{\phi}(x) - \meas{\phi}(x_{0}))(x-x_{0})^{T}}{\|x-x_{0}\|^{2}}.
\]
Then we can write 
\begin{align*}
\Phi_{\rm p}(x,U_{\rm p}) - \Phi_{\rm p}(x_{0},U_{\rm p})&=\Psi_{\rm p}(x,x_{0},U_{\rm p})(x-x_{0})
\end{align*}
where
\[
\resizebox{\columnwidth}{!}{$\Psi_{\rm p}(x,x_{0},U_{\rm p}) = \begin{bmatrix}
    M(x,x_0) \\
    M(x - f( u_{0}),x_{0} - f( u_{0})) \\
    \vdots\\\displaystyle
    M\left(x - \sum_{i=-p+2}^{0}f(u_{i}),x_{0} - \sum_{i=-p+2}^{0}f(u_{i})\right)  
    \end{bmatrix},$}.
\]
Using a similar decomposition of $\phi$, we can write
\[
\Phi_{\rm f}(x,U_{\rm f}) - \Phi_{\rm f}(x_{0},U_{\rm f}) = \Psi_{\rm f}(x,x_0,U_{\rm f})(x-x_{0}).
\]
\end{lemma}

\begin{proof}
The proof is by straightforward substitution and is omitted.
\end{proof}

We can now state the main result of the paper.

\begin{theorem}\label{thm:mainresult}
Given \eqref{eqn:altpast}, \eqref{eqn:altfuture}, \eqref{eqn:altestimate}, and \eqref{eqn:altpredict}, with $U_{\rm p}$ and $U_{\rm f}$ fixed, $g_{k}$ the impulse response of an $n$ dimensional system, $p > n$, and ${\rm Null}(\meas{T}_{1}^{p,r})\subseteq {\rm Null}(T_{p+1}^{p_{f},r})$, let $\Psi_{\rm p}$ and $\Psi_{\rm f}$ be defined as in Lemma \ref{lem:MPsi}. Let $L_{1} = \sup_{x}\|\Psi_{\rm p}(x,x_{0},U_{\rm p})\|$, $L_{2} = \sup_{x}\|\Psi_{\rm f}(x,x_{0},U_{\rm f})\|$, and $\alpha =  \inf_{x}\sigma_{\rm min}((I-\meas{T}_{1}^{p,r}(\meas{T}_{1}^{p,r})^{+})\Psi_{\rm p}(x,x_{0},U_{\rm p}))$. The  bound 
\begin{equation}\label{eqn:mainbound}
\resizebox{\columnwidth}{!}{$\|\est{Y}_{\rm f} - Y_{\rm f}\|\leq \left(\left(1+\frac{L_{1}}{\alpha}\right)\sigma_{\rm max}(T_{p+1}^{p_{f},r}(\meas{T}_{1}^{p,r})^{+})+\frac{L_{2}}{\alpha}\right) \|\est{\meas{Y}}_{\rm p} - \meas{Y}_{\rm p}\|$}
\end{equation}
can be established under either of the following conditions.
\begin{itemize}
    \item Suppose $\meas{T}_{2}^{n,n}(\meas{T}_{1}^{n,n})^{+}$ has eigenvalues satisfying $\lambda \neq 1$. Then there exists $\gamma_{0} >0$ such that for any sector bounded function $\meas{\phi}(x)$ with a $\gamma$-approximation $B$ where $\gamma  < \gamma_{0}$ and $B$ has full column rank, the bound~\eqref{eqn:mainbound} holds with $\alpha>0$.
    \item Suppose $\meas{T}_{2}^{n,n}(\meas{T}_{1}^{n,n})^{+}$ has eigenvalues each satisfying $|\lambda - 1|>\rho$. Let $x_{k}$ be the trajectory of the equilibrium state. If $\meas{\phi}$ is $L$-smooth and the Jacobian sequence $J(x_{k})$ is band-limited with bound $\rho$ and full column rank for all~$k$, then there exists $\epsilon_{0}$ such that when $\|\est{x}_{0}-x_{0}\|<\epsilon_{0}$ the bound~\eqref{eqn:mainbound} holds with $\alpha>0$.
\end{itemize}
\end{theorem}
The proof of this theorem is the subject of the remainder of the paper. The bound is similar to Theorem \ref{thm:maintheoremlinear}, but has additional terms related to the nonlinear equilibrium function. The theorem provides a set of observability conditions for the IREM: in addition to ${\rm Null}(\meas{T}_{1}^{p,r})\subseteq {\rm Null}(T_{p+1}^{p_{f},r})$, the eigenvalues of $\meas{T}_{2}^{n,n}(\meas{T}_{1}^{n,n})^{+}$ are restricted to not include $1$ or a region around $1$, depending on the form of the nonlinearity. The non-zero eigenvalues of  $\meas{T}_{2}^{n,n}(\meas{T}_{1}^{n,n})^{+}$ are the same as the eigenvalues of a state-space representation (see the proof of Lemma \ref{lem:impulserepresentation}), thus this requirement implies that the impulse response part of the representation contains no integrators, and all integration action is part of the equilibrium model.  

\section{Proofs of Theorems \ref{thm:maintheoremlinear} and \ref{thm:mainresult}}\label{sec:proofs}

The proofs of the main results are completed in this section. The first lemma provides a relationship between the solutions to two different algebraic relationships with common variables. 

\begin{lemma}\label{lem:mainlemma}
    Given algebraic relationships
    \begin{align}
    Z_{1} = M_{1} X + \Phi_{1}(w) \label{eqn:nonlinrel1}, \\
    Z_{2} = M_{2} X + \Phi_{2}(w)\label{eqn:nonlinrel2}.
    \end{align}
    where $\Phi_{1}$ and $\Phi_{2}$ can be represented at each $w$ as $\Phi_{i}(w) = \Psi_{i}(w)w$.
    Let $L_{1} = \sup_{w}\|\Psi_{1}(w)\|$, $L_{2} = \sup_{w}\|\Psi_{2}(w)\|$,
    and
    \[
    \alpha = \inf_{w}\sigma_{\rm min}((I-M_{1}M_{1}^{+})\Psi_{1}(w)).
    \]
    \begin{itemize}
    \item If $w=0$, and ${\rm Null}(M_1)\subseteq {\rm Null}(M_2)$ then 
    \[
    \|Z_{2}\| \leq \sigma_{\rm max}(M_{2}M_{1}^+) \|Z_{1}\|. 
\]
    \item If $\alpha>0$ and ${\rm Null}(M_1)\subseteq {\rm Null}(M_2)$, then 
    \[
    \|Z_{2}\| \leq \left((1+(L_{1}/\alpha))\sigma_{\rm max}(M_{2}M_{1}^+) + L_{2}/\alpha\right)\|Z_{1}\|. 
\]
\end{itemize}
\end{lemma}
\begin{proof} Multiplying by $M_{1}^{T}$ on the left and adding $\lambda X $ on both sides of \eqref{eqn:nonlinrel1}
\[
M_1^{T}Z_{1}  + \lambda X = (M_{1}^{T}M_{1} + \lambda I)X + M_{1}^{T}\Phi_1(w).
\]

Taking $\lambda >0$, the matrix $(M_{1}^{T}M_{1} + \lambda I)$ is invertible. Multiplying by the inverse on the left and rearranging
\begin{equation}\label{eqn:solveforX}
X = (M_{1}^{T}M_{1} + \lambda I)^{-1}M_1^{T}(Z_{1} - \Phi_1(w)) + \lambda (M_{1}^{T}M_{1} + \lambda I)^{-1}X. 
\end{equation}
Plugging \eqref{eqn:solveforX} into \eqref{eqn:nonlinrel2},
\begin{multline*}
Z_{2} = M_{2}(M_{1}^{T}M_{1} + \lambda I)^{-1}M_1^{T}(Z_{1} - \Phi_{1}(w)) + \\\lambda M_{2}(M_{1}^{T}M_{1} + \lambda I)^{-1}X + \Phi_{2}(w).
\end{multline*}
We now consider two cases.

\textbf{Case 1: $w=0$.} After taking the limit as $\lambda \searrow 0$ and substituting $\Phi_{1}(w) = \Phi_{2}(w)=0$,  we use Lemma \ref{lem:pseduoinverse}, and note that since ${\rm Null}(M_{1}) \subseteq {\rm Null}(M_{2})$, $\sigma_{\rm max}( M_{2}(M_{1}^{T}M_{1} + \lambda I)^{-1})$ is bounded for all $\lambda>0$, and  $\lim_{\lambda \searrow 0}\lambda M_{2}(M_{1}^{T}M_{1} + \lambda I)^{-1} =0$, so that
\[
Z_{2} = M_{2}M_{1}^{+}Z_{1}.
\]
Thus
\[
\|Z_{2}\| \leq \sigma_{\rm max}(M_{2}M_{1}^{+})\|Z_{1}\|.
\]

\textbf{Case 2: $\alpha>0$.} Plug \eqref{eqn:solveforX} into \eqref{eqn:nonlinrel1},
\begin{multline*}
Z_{1} = M_{1}(M_{1}^{T}M_{1} + \lambda I)^{-1}M_1^{T}(Z_{1} - \Phi_{1}(w)) + \\\Phi_{1}(w) + \lambda M_{1}(M_{1}^{T}M_{1} + \lambda I)^{-1}X.
\end{multline*}
Taking the limit as $\lambda \searrow 0$, and noting that the term on the right converges to zero,
\[
Z_{1} = M_{1}M_{1}^{+}(Z_{1} - \Phi_{1}(w)) + \Phi_{1}(w).
\]
Rearranging and substituting for $\Phi_{1}(w)$,
\[
(I-M_{1}M_{1}^{+})Z_{1} = (I-M_{1}M_{1}^{+})\Psi_{1}(w)w.
\]
By assumption, $\sigma_{\rm min}((I-M_{1}M_{1}^{+})\Psi_{1}(w))\geq \alpha$, thus
\begin{equation*}
\|w\| \leq \|Z_{1}\|/\alpha.
\end{equation*}
so that 
\begin{equation}\label{eqn:boundonw}
\|\Phi_{1}(w)\| \leq \|\Psi_{1}(w)\|\frac{\|Z_{1}\|}{\alpha}\leq \frac{L_{1}}{\alpha}\|Z_{1}\|,
\end{equation}
and similar for $\Phi_{2}(w)$.

Thus
\begin{multline*}
\|Z_{2}\| \leq \sigma_{\rm max}( M_{2}(M_{1}^{T}M_{1} + \lambda I)^{-1} M_{1}^{T})\left(\|Z_{1}\| +\frac{L_{1}}{\alpha}\|Z_{1}\|\right) + \\
 \lambda\sigma_{\rm max}( M_{2}(M_{1}^{T}M_{1} + \lambda I)^{-1})\|X\|  + \frac{L_{2}}{\alpha}\|Z_{1}\|.
\end{multline*}
Taking the limit as $\lambda \searrow 0$, using Lemma \ref{lem:pseduoinverse}, \eqref{eqn:boundonw} and noting that since ${\rm Null}(M_1)\subseteq {\rm Null}(M_2)$ the term involving $\|X\|$ goes to zero, we arrive at
\[
\|Z_{2}\| \leq \left((1+L_{1}/\alpha)\sigma_{\rm max}(M_{2}M_{1}^+) + L_{2}/\alpha\right)\|Z_{1}\|. 
\]
\end{proof}

This lemma can be immediately applied to state the proof of Theorem \ref{thm:maintheoremlinear}.

{\em Proof of Theorem: \ref{thm:maintheoremlinear}.}
Subtracting \eqref{eqn:past} from  \eqref{eqn:estimate}, and \eqref{eqn:future} from \eqref{eqn:predict},
\begin{align*}
\est{\meas{Y}}_{\rm p} - \meas{Y}_{\rm p} &= \meas{T}_{1}^{p,r}(\est{U}^{\rm e} - U^{\rm e}).\\
\est{Y}_{\rm f} - Y_{\rm f} &= T_{p+1}^{p_{f},r}(\est{U}^{\rm e} - U^{\rm e}).
\end{align*}
Using Lemma \ref{lem:mainlemma} with $w=0$, $X = \est{U}^{\rm e} - U^{\rm e}$, $Z_{1} = \est{\meas{Y}}_{\rm p} - \meas{Y}_{\rm p}$, $Z_{2} = \est{Y}_{\rm f}-Y_{\rm f}$, $M_{1} = \meas{T}_{1}^{p,r}$, and $M_{2} = T_{p+1}^{p_{f},r}$, the result is proven.
\hfill \qed\vspace{.1in}

More work is needed to prove Theorem \ref{thm:mainresult}.  
In the following theorem, Lemma \ref{lem:mainlemma} is used to obtain the bounds expressed in Theorem \ref{thm:mainresult}, but with a different set of assumptions. 

\begin{theorem}\label{thm:maintheorem}
Given \eqref{eqn:altpast}, \eqref{eqn:altfuture}, \eqref{eqn:altestimate}, and \eqref{eqn:altpredict} with $x_{0}$, $U_{\rm p}$ and $U_{\rm f}$ fixed. Find $\Psi_{\rm p}(x,x_{0},U_{\rm p})$ and $\Psi_{\rm f}(x,x_{0},U_{\rm f})$ such that
\begin{align*}
\Phi_{\rm p}(x,U_{\rm p}) - \Phi_{\rm p}(x_{0},U_{\rm p}) & = \Psi_{\rm p}(x,x_{0},U_{\rm p})(x-x_{0})  \\
\Phi_{\rm f}(x,U_{\rm f}) - \Phi_{\rm f}(x_{0},U_{\rm f}) & = \Psi_{\rm f}(x,x_{0},U_{\rm f})(x-x_{0}).   
\end{align*}
Let
\[
\alpha = \inf_{x}\sigma_{\rm min}((I-\meas{T}_{1}^{p,r}(\meas{T}_{1}^{p,r})^{+})\Psi_{\rm p}(x,x_{0},U_{\rm p})),
\]
$L_{1} = \sup_{x}\|\Psi_{\rm p}(x,x_{0},U_{\rm p})\|$, and $L_{2} = \sup_{x}\|\Psi_{\rm f}(x,x_{0},U_{\rm f})\|$.
If ${\rm Null}\left(\meas{T}_{1}^{p,r}\right) \subseteq {\rm Null}\left(T_{p+1}^{p_{f},r}\right)$ and $\alpha >0$, then
\[
\resizebox{\columnwidth}{!}{$
\|\est{Y}_{\rm f} - Y_{\rm f}\|\leq \left(\left(1+\frac{L_{1}}{\alpha}\right)\sigma_{\rm max}(T_{p+1}^{p_{f},r}(\meas{T}_{1}^{p,r})^{+})+\frac{L_{2}}{\alpha}\right) \|\est{\meas{Y}}_{\rm p} - \meas{Y}_{\rm p}\|.$}
\]
\end{theorem}

\begin{proof} Subtracting common terms,
\begin{align*}
\est{\meas{Y}}_{\rm p} - \meas{Y}_{\rm p} &= \meas{T}_{1}^{p,r}(\est{U}^{\rm e} - U^{\rm e}) + \Phi_{\rm p}(\est{x}_{0},U_{\rm p}) - \Phi_{\rm p}(x_{0},U_{\rm p})\\
\est{Y}_{\rm f} - Y_{\rm f}& = T_{p+1}^{p_{f},r}(\est{U}^{\rm e} - U^{\rm e}) + \Phi_{\rm f}(\est{x}_{0},U_{\rm f}) - \Phi_{\rm f}(x_{0},U_{\rm f}). 
\end{align*}
The result then follows from Lemma \ref{lem:mainlemma} with $w=\est{x}_{0}-x_{0}$, $\Psi_{1}(w) = \Psi_{\rm p}(w+x_{0},x_{0},U_{\rm p})$, $\Psi_{2}(w) = \Psi_{\rm f}(w+x_{0},x_{0},U_{\rm f})$, $X = \est{U}^{\rm e} - U^{\rm e}$, $Z_{1} = \est{\meas{Y}}_{\rm p} - \meas{Y}_{\rm p}$, $Z_{2} = \est{Y}_{\rm f}-Y_{\rm f}$, $M_{1} = \meas{T}_{1}^{p,r}$, and $M_{2} = T_{p+1}^{p_{f},r}$.
\end{proof}

What remains is to show that the assumptions of Theorem \ref{thm:mainresult} lead to the assumptions of Theorem \ref{thm:maintheorem}. To this, we need to develop a few other tools. The following lemma describes a decomposition of the impulse response of a finite-dimensional system.

\begin{lemma}\label{lem:impulserepresentation}
 If $g_{k}$ is the impulse response of an $n$ dimensional linear time-invariant system, then for $k\geq 1$
\begin{align}\label{eqn:impulsedecomposition}
g_{k} &= \sum_{i=1}^{m}\sum_{\ell=0}^{\ell_{i}}K_{i,\ell}k^{\ell}\lambda_{i}^k + \sum_{i=1}^{\ell_{0}}L_{i}\delta(k-i)
\end{align}
where $K_{i,\ell}$, $L_{i}\in\mathbb{R}^{q\times m}$ are constants, $\lambda_{i}\in \mathbb{C}$ are the eigenvalues of $T_{2}^{n,n}(T_{1}^{n,n})^{+}$ not equal to zero, $\ell_i= \eta_{i}-\rho_{i}$ where $\eta_{i}$ is the algebraic multiplicity of eigenvalue $\lambda_{i}$ and $\rho_{i}$ is the geometric multiplicity, and $\ell_{0}= n-\sum_{i=1}^{m}(\eta_{i}+1)$. 
\end{lemma}
\begin{proof} A minimal state-space realization of the impulse response for $k\geq 1$ can be found using the Ho-Kalman algorithm (i.e. $A$, $B$, and $C$ such that $g_{k} = CA^{k-1}B$). The decomposition \eqref{eqn:impulsedecomposition} is then easily derived for a state-space system where $A$ has eigenvalues $\lambda_{i}$ and the requisite algebraic and geometric multiplicity \cite{chen1984linear}. $\ell_{0}$ will be equal to the number of eigenvalues of $A$ at zero. It remains to show that $T_{2}^{n,n}(T_{1}^{n,n})^{+}$ has the same non-zero eigenvalues as a minimal state-space realization. A state-space realization is related to $T_{2}^{n,n}$ and $T_{1}^{n,n}$ via $T_{1}^{n,n} = \mathcal{O}^{n}\mathcal{C}$, and $T_{2}^{n,n} = \mathcal{O}^{n}A\mathcal{C}$, using the same notation as in the proof of Theorem \ref{thm:ssobservability}. If the realization is minimal, $\mathcal{O}^{n}$ is full column rank and $\mathcal{C}$ is full row rank so that $(\mathcal{O}^{n})^{+}\mathcal{O}^{n} = I$, and $\mathcal{C}\mathcal{C}^{+}=I$. This also implies $(\mathcal{O}^{n}\mathcal{C})^{+} = \mathcal{C}^{+}(\mathcal{O}^{n})^{+}$. 
Thus
\[
T_{2}^{n,n}(T_{1}^{n,n})^{+} = \mathcal{O}^{n}A\mathcal{C}\mathcal{C}^{+}(\mathcal{O}^{n})^{+}.
\]
 Let $r$ be the number of rows of $\mathcal{O}^{n}$. Then
\begin{align*}
\det(\lambda I_{r} - T_{2}^{n,n}(T_{1}^{n,n})^{+}) &= \det(\lambda I_{r} - \mathcal{O}^{n}A(\mathcal{O}^{n})^{+}), \\
& = \lambda^{r-n}\det(\lambda I_{n} - (\mathcal{O}^{n})^{+}\mathcal{O}^{n}A),\\
& = \lambda^{r-n}\det(\lambda I_{n} - A).
\end{align*}
In the second step we used Sylvester's determinant theorem ($\det(X+AB) = \det(X)\det(I+BX^{-1}A)$), along with $\det(aB) = a^{m}\det(B)$ for $B\in\mathbb{R}^{m\times m}$. Thus the characteristic equation for $T_{2}^{n,n}(T_{1}^{n,n})^{+}$ contains the same roots as the characteristic equation for $A$, plus additional roots at zero.
\end{proof}

The next set of results, culminating in Theorem  \ref{thm:hindependent}, establishes when a matrix made up of terms of the form $k^{\ell}\lambda^{k}$ is full column rank. One well known way to prove the independence of the functions $\lambda^{k}$ with different $\lambda$ is to utilize the shift property $\lambda^{k+k_{0}} =\lambda^{k_0}\lambda^{k}$ (i.e. the time shifted function is a scaled version of itself). This approach is adapted here to show that $k^{\ell_{0}}\lambda^{k}$ as columns of a matrix are linearly independent. The functions $k^{\ell_{0}}\lambda^{k}$ have a related shift property, where the shifted function can be written as a linear combination of $k^{\ell}\lambda^{k}$ with $\ell\leq \ell_{0}$ but with the same $\lambda$.

\begin{lemma}\label{lem:remap}
Let
\[
h({k})_{(\lambda,\ell_{0})} = k^{\ell_{0}}\lambda^{k},
\]
where $\ell_{0}$ is a nonnegative integer.
For any $\Delta$
\[
h(k+\Delta)_{(\lambda,\ell_{0})} = \sum_{\ell=0}^{\ell_{0}}b_{\ell}h({k})_{(\lambda,\ell)}
\]
with 
\[
b_{\ell} = \binom{\ell_{0}}{\ell}\Delta^{\ell_{0}-\ell}\lambda^{\Delta}. 
 \]
\end{lemma}
Proof: Using the binomial expansion, we have
\begin{align*}
h(k+\Delta)_{(\lambda,\ell_{0})} &= (k+\Delta)^{\ell_{0}}\lambda^{k+\Delta} \\ &= \left( \sum_{\ell=0}^{\ell_{0}} \binom{\ell_{0}}{\ell} \Delta^{\ell_{0}-\ell} k^\ell \right)\lambda^{k+\Delta} \\ &= 
\sum_{\ell=0}^{\ell_{0}} \underbrace{ \binom{\ell_{0}}{\ell} \Delta^{\ell_{0}-\ell} \lambda^{\Delta}}_{b_{\ell}}\underbrace{\left.\rule[-9pt]{0pt}{12pt}k^\ell \lambda^{k}\right.}_{h(k)_{(\lambda,\ell)}}.
\end{align*}
\hfill \qed

A preliminary result shows that a matrix with columns of $k^{\ell}\lambda^{k}$ with the same $\lambda$ and different $\ell$ are linearly independent.

\begin{lemma} \label{lem:vandemonde}
   Let $h({k})_{(\lambda,i)} = k^{i}\lambda^{k}$. For $\lambda \neq 0$, integers $1 \leq k_1\leq k_2$ and any collection of sorted integers $i_{1},\cdots, i_{m}$, a  matrix of the form
    \[
    M = \begin{bmatrix}h(k_{2})_{(\lambda,i_{1})} & \cdots 
    & h(k_{2})_{(\lambda,i_{m})} \\
    h(k_{2}-1)_{(\lambda,i_{1})} & \cdots 
    & h(k_{2}-1)_{(\lambda,i_{m})} \\
    \vdots & \vdots & \vdots \\
    h(k_{1})_{(\lambda,i_{1})} & \cdots &  h(k_{1})_{(\lambda,i_{m})}\end{bmatrix}
    \]
    with $k_{2}-k_{1}+1 \geq i_m-i_{1}-1$ is full column rank.
    \end{lemma}
Proof: We have
\[
\resizebox{\columnwidth}{!}{$\begin{aligned}
M& = \begin{bmatrix} k_{2}^{i_{1}}\lambda^{k_{2}} & \cdots & k_{2}^{i_{m}}\lambda^{k_{2}} \\
(k_{2}-1)^{i_{1}}\lambda^{k_{2}-1} & \cdots & (k_{2}-1)^{i_{m}}\lambda^{k_{2}-1} \\
    \vdots & \vdots & \vdots \\
(k_{1})^{i_{1}}\lambda^{k_{1}} & \cdots & (k_{1})^{i_{m}}\lambda^{k_{1}}\end{bmatrix}
\\
& = \operatorname{diag}\begin{bmatrix} \lambda^{k_{2}}k_{2}^{i_{1}}\\
\lambda^{k_{2}-1}(k_{2}-1)^{i_{1}} \\
    \vdots   \\
\lambda^{k_{1}}(k_{1})^{i_{1}} 
\end{bmatrix}\begin{bmatrix} 1 &k_{2}^{i_{2}-i_{1}}& \cdots & k_{2}^{i_{m}-i_{1}} \\
 1 & (k_{2}-1)^{i_{2}-i_{1}} & \cdots & (k_{2}-1)^{i_{m}-i_{1}} \\
    \vdots & \vdots & \vdots \\
1  & k_{1}^{i_{2}-i_{1}} & \cdots & (k_{1})^{i_{m}-i_{1}}\end{bmatrix}
\end{aligned}$}
\]
Since $\lambda\neq 0$ the matrix on the left is clearly invertible. The matrix on the right contains columns of a $k_{2}-k_{1}+1$ by $k_{2}-k_{1}+1$ Vandermonde matrix, which is invertible, so the columns are linearly independent. Thus $M$ is full column rank.\hfil \qed

The following theorem uses the shift property from Lemma \ref{lem:remap} along with Lemma \ref{lem:vandemonde} to show a matrix made up of $k^{\ell}\lambda^{k}$ with different $\lambda$ is linearly independent, via proof by contradiction. 

\begin{theorem}\label{thm:hindependent}
    Given unique $\lambda_{i},\ell_{i}$, $i=1,\cdots,m$, $\lambda_{i}\neq 0$, and integers $1\leq k_{1}\leq k_{2}$. Let $h({k})_{(\lambda_{i},j)} = k^{j}\lambda_{i}^{k}$ for $i=1,\cdots,m$ and $j=0,\cdots,\ell_{i}$.  A matrix of the form
    \[
    \resizebox{\columnwidth}{!}{$
    M = \begin{bmatrix}h(k_{2})_{(\lambda_{1},0)} & \cdots 
    & h(k_{2})_{(\lambda_{1},\ell_{1})} & h(k_{2})_{(\lambda_{2},0)}& \cdots &  h(k_{2})_{(\lambda_{m},\ell_{m})} \\
    h(k_{2}-1)_{(\lambda_{1},0)} & \cdots 
    & \cdots & \cdots &\cdots &  h(k_{2}-1)_{(\lambda_{m},\ell_{m})} \\
    \vdots & \vdots & \vdots & \vdots& \vdots& \vdots  \\
    h(k_{1})_{(\lambda_{1},0)} & \cdots & h(k_{1})_{(\lambda_{1},\ell_{1})} & h(k_{1})_{(\lambda_{2},0)} &\cdots &  h(k_{1})_{(\lambda_{m},\ell_{m})} \end{bmatrix}$}
    \]
    with $k_{2}-k_{1}+1 \geq \sum_{i=1}^{m}(\ell_{i}+1)$ is full column rank.
\end{theorem}
Proof: Suppose not. Then there exists
\[
c = \begin{bmatrix} c_{\lambda_{1},0} & \cdots & c_{\lambda_{1},\ell_1} & \cdots & c_{\lambda_{m},\ell_{m}}\end{bmatrix}^{T}
\]
such that $Mc=0$, $c\neq 0$. We assume that the columns of $M$ with nonzero elements of $c$ are associated with at least two different eigenvalues, for if not, these columns comprise a matrix with the form given in Lemma \ref{lem:vandemonde}, which is full column rank, a contradiction. 

Let $\bar{\ell}_{i}$ be the largest value of $\ell$ with non-zero elements of $c$, for each $i$. For a matrix $M$, let $M(i:j,:)$ be the submatrix containing the $i$th through $j$th row, and let $p$ be the last row of $M$. By Lemma \ref{lem:remap}, with $\Delta =1$,  $M(2:p,:)c = M(1:p-1,:)\bar{c}$, where 
\[
 \bar{c}_{i,\ell} = \left(\lambda_{i}\sum_{s=\ell}^{\bar{\ell}_{i}}\binom{s}{\ell} c_{i,s} \right).
\]
Note that we have both $M(1:p-1,:)c = 0$ and $M(1:p-1,:)\bar{c}=0$.
Also note that that for all $i$,
\begin{equation}\label{eqn:coefrelationship}
\bar{c}_{i,\bar{\ell_{i}}} = \lambda_{i}c_{i,\bar{\ell_{i}}}.
\end{equation}
Choose $i$ such that $c_{i,\bar{\ell}_{i}}\neq 0$. Because of \eqref{eqn:coefrelationship}, $\bar{c}/\lambda_{i} - c$ has a zero at the index associated with $c_{i,\bar{\ell}_{i}}$. In addition, we have
\[
M(1:p-1,:)(\bar{c}/\lambda_{i} - c) = 0. 
\]
It remains to show that $(\bar{c}/\lambda_{i} - c)\neq 0$. Since there is a non-zero coefficient associated with a different eigenvalue, $\lambda_{j}$, then
\[
\bar{c}_{j,\bar{\ell_{j}}}/\lambda_{i} - c_{j,\bar{\ell_{j}}} = \left(\frac{\lambda_{j}}{\lambda_{i}}-1\right) c_{j,\bar{\ell}_{j}} 
\]
and since $\lambda_{j}\neq \lambda_{i}$ this is non zero. 

Redefine $c$  as $(\bar{c}/\lambda_{i} - c)$. Let $n$ be the index of $c$ associated with $c_{i,\bar{\ell}_{i}}$. Remove row $p$ and column $n$ from  $M$ and element $n$ from $c$, and note that for this new $M$ and $c$, we have $Mc=0$ and $c\neq 0$, thus we can repeat this process again. Continue until columns associated with only one eigenvalue remain (which will take strictly fewer than $\sum_{i=1}^{m}(\ell_{i}+1)$ iterations, so $M$ will have at least one row). By Lemma \ref{lem:vandemonde} this matrix is full column rank, which contradicts  $Mc=0$.\hfill\qed\\

A last technical result bounds the approximation error for a first order Taylor series when the function is $L$-smooth.

\begin{lemma}\label{lem:Lsmooth}
If $\meas{\phi}:\mathbb{R}^{s}\rightarrow \mathbb{R}^{q_m}$ is $L$-smooth, then we can write its first-order Taylor series expansion around a point $x_0 \in \mathbb{R}^{s}$ as
\begin{equation}
\meas{\phi}(x) = \meas{\phi}(x_0) + J(x_0) (x-x_0) + r(x,x_0)
\label{eq:taylorJr}
\end{equation}
where $r(x,x_0)$ is a $q_m \times 1$ residual vector that satisfies  
\begin{equation}
\| r(x,x_0) \| \le \frac{L}{2} \| x-x_0 \|^2.
\label{eq:resbound}
\end{equation}
\end{lemma}

\begin{proof}
    From the Fundamental Theorem of Calculus, we can write
    $$
    \meas{\phi}(x) - \meas{\phi}(x_0) = \int_0^1 J(x_0 + t(x - x_0)) (x - x_0) dt.
    $$
    We can isolate $r(x, x_0)$ by subtracting the linear approximation $J(x_0)(x-x_0)$ from both sides:
    $$
    r(x, x_0) = \int_0^1 \left( J(x_0 + t(x-x_0)) - J(x_0) \right) (x-x_0) dt.
    $$
    Taking the norm of both sides and applying the triangle inequality with~\eqref{eq:Lsmooth} yields the bound:
    \begin{align*}
    \|r(x,x_0)\| 
    &=
    \left\|  \int_0^1 \left( J(x_0 + t(x-x_0)) - J(x_0) \right) (x-x_0) dt \right\| \\
    &\le \int_0^1 \| \left( J(x_0 + t(x-x_0)) - J(x_0) \right) (x-x_0) \| dt \\
    &\le \int_0^1 \| J(x_0 + t(x-x_0)) - J(x_0) \| \| x-x_0 \| dt \\
    &\le \int_0^1 L t \| x-x_0 \|^2 dt \\
    &= \frac{L}{2} \|x-x_0\|^2.
    \end{align*}
\end{proof}

Theorem \ref{thm:finaltheorem} shows that the assumptions of Theorem \ref{thm:mainresult} lead to the assumptions of Theorem \ref{thm:maintheorem}, and provides the proof of Theorem \ref{thm:mainresult}.

\begin{theorem}\label{thm:finaltheorem}
Given \eqref{eqn:altpast}, \eqref{eqn:altfuture}, \eqref{eqn:altestimate}, and \eqref{eqn:altpredict} with $x_{0}$, $U_{\rm p}$ and $U_{\rm f}$ fixed. There exists $\Psi_{\rm p}$ such that $\Phi_{\rm p}(x,U_{\rm p}) - \Phi_{\rm p}(x_{0},U_{\rm p})  = \Psi_{\rm p}(x,x_{0},U_{\rm p})(x-x_{0})$ and
\[
\alpha = \inf_{x}\sigma_{\rm min}((I-\meas{T}_{1}^{p,r}(\meas{T}_{1}^{p,r})^{+})\Psi_{\rm p}(x,x_{0},U_{\rm p}))>0
\]
under either of the following conditions. 
\begin{itemize}
    \item Suppose $\meas{T}_{2}^{n,n}(\meas{T}_{1}^{n,n})^{+}$ has eigenvalues satisfying $\lambda \neq 1$. Then there exists $\gamma_{0} >0$ such that for any sector bounded function $\meas{\phi}(x)$ with a $\gamma$-approximation $B$ where $\gamma  < \gamma_{0}$ and $B$ has full column rank, then $\alpha>0$.
    \item Suppose $\meas{T}_{2}^{n,n}(\meas{T}_{1}^{n,n})^{+}$ has eigenvalues each satisfying $|\lambda-1|>\rho$. Let $x_{k}$ be the trajectory of the equilibrium state. If $\meas{\phi}$ is $L$-smooth and the Jacobian sequence $J(x_{k})$ is band-limited with bound $\rho$, and full column rank for all $k$, then there exists $\epsilon_{0}$ such that  $\alpha>0$ when $\|\est{x}_{0}-x_{0}\|<\epsilon_{0}$.
\end{itemize}
\end{theorem}
\begin{proof} By Lemma \ref{lem:impulserepresentation}, the impulse response $\meas{g}_{k}$ can be written as
\begin{align}
\meas{g}_{k} &= \sum_{i=1}^{m}\sum_{\ell=0}^{\ell_{i}}K_{i,\ell}k^{\ell}\lambda_{i}^k+ \sum_{i=1}^{\ell_{0}}L_{i}\delta(k-i) \\
& = \sum_{i=1}^{m}\sum_{\ell=0}^{\ell_{i}}K_{i,\ell}h(k)_{(\lambda_i,\ell)}+ \sum_{i=1}^{\ell_{0}}L_{i}\delta(k-i)
\end{align}
where all $\lambda_{i}\neq 1$.
Let
\[
H(k)_{(\lambda,\ell)} = \begin{bmatrix} h(k)_{(\lambda,0)}I &h(k)_{(\lambda,1)}I& \cdots  &h(k)_{(\lambda,\ell)}I \end{bmatrix}
\]
and
\[
K_{i} = \begin{bmatrix} K_{i,0} \\ \vdots \\ K_{i,\ell_{i}}\end{bmatrix}.
\]
Then
\[
\resizebox{\columnwidth}{!}{$\begin{bmatrix} \meas{g}_{p} \\ \meas{g}_{p-1}\\\vdots \\\meas{g}_{1}\end{bmatrix} =\underbrace{\begin{bmatrix}H(p)_{(\lambda_1,\ell_{1})} &  \cdots & H(p)_{(\lambda_{m},\ell_m)}\\
H(p-1)_{(\lambda_1,\ell_{1})} &  \cdots & H(p-1)_{(\lambda_{m},\ell_m)} \\
\vdots & & \vdots
\\ H(1)_{(\lambda_1,\ell_{1})} &  \cdots & H(1)_{(\lambda_m,\ell_{m})}\end{bmatrix}}_{\mathcal{H}}\begin{bmatrix} K_{1}  \\
K_{2}  \\
\vdots  \\
K_{m} 
\end{bmatrix} + \begin{bmatrix} 0 \\ 0 \\ \vdots \\ L_{1}\end{bmatrix}.$}
\]

Using Lemma \ref{lem:remap}, 
\[
  H(j+k)_{(\lambda,\ell)} = H(j)_{(\lambda,\ell)}\Gamma(k,\lambda,\ell)
  \]
where
\[
    \Gamma(k,\lambda,\ell)  =\begin{bmatrix} \begin{matrix}\begin{matrix}\left[C(k,\lambda,0)\right] \\ 0 \end{matrix} \\ \vdots \\ 0\end{matrix}& \begin{matrix} \left[\rule{0pt}{15pt}C(k,\lambda,1)\right]\\ \vdots \\ 0\end{matrix}& \cdots & \left[\rule{0pt}{30pt}C(k,\lambda,\ell)\right]\end{bmatrix}
\]
and 
\[
C(k,\lambda,\ell) = \begin{bmatrix} \binom{\ell}{0}k^{\ell}\lambda^{k}I \\ \binom{\ell}{1}k^{\ell-1}\lambda^{k}I \\ \vdots \\ \binom{\ell}{\ell}\lambda^{k}I
\end{bmatrix}.
\]
Noting that each column $\meas{T}_{1}^{p,r}$ is a time shifted version of the first column, we can write
\begin{align}\label{eqn:maindecomposition}
\meas{T}_{1}^{p,r} &= \begin{bmatrix} g_{p} & g_{p+1} & \cdots & g_{p+r-1} \\
g_{p-1} & g_{p} &\cdots & g_{p+r-2} \\
\vdots  & \vdots &  & \vdots \\
g_{1} & g_{2} & \cdots & g_{r}\end{bmatrix}=\mathcal{H} K + L,
\end{align}
where
\begin{align*}
K &= \begin{bmatrix}K_{1} & \Gamma(1,\lambda_{1},\ell_{1})K_{1} & \cdots & \Gamma(r,\lambda_{1},\ell_{1})K_{1} \\ 
K_{2} & \Gamma(1,\lambda_{2},\ell_{2})K_{2} & \cdots & \Gamma(r,\lambda_{2},\ell_{2})K_{2}\\
\vdots & \vdots  &  & \vdots\\
K_{m} & \Gamma(1,\lambda_{m},\ell_{m})K_{m} & \cdots & \Gamma(r,\lambda_{m},\ell_{m})K_{m}
\end{bmatrix}
\end{align*}
and
\[
L = \begin{bmatrix} 0 & 0 & \cdots & 0\\
\vdots & \vdots & & \vdots \\
L_{\ell_{0}} & 0 & \cdots & 0\\

\vdots & \ddots &  & \vdots \\
 L_{1} & L_{2} & \cdots & L_{r}
 \end{bmatrix}.
\]
Above, if $r > \ell_0$, some of the rightmost columns of $L$ will be all zeros.  

Now,  
\begin{align*}
\alpha &= \inf_{x}\sigma_{\rm min}((I-\meas{T}_{1}^{p,r}(\meas{T}_{1}^{p,r})^{+})\Psi_{\rm p}(x,x_{0},U_{\rm p}))\\
&=\inf_{\|\psi\|=1}\|(I-\meas{T}_{1}^{p,r}(\meas{T}_{1}^{p,r})^{+})\Psi_{\rm p}(x,x_{0},U_{\rm p})\psi\|  \\
& = \inf_{x,\theta,\|\psi\|=1}\|\Psi_{\rm p}(x,x_{0},U_{\rm p})\psi-\meas{T}_{1}^{p,r}\theta\|\\
& = \inf_{x,\theta,\|\psi\|=1}\|\Psi_{\rm p}(x,x_{0},U_{\rm p})\psi-(\mathcal{H}K+ L)\theta\|.
\end{align*}

This optimization problem has a positive minimum if the column space of $\Psi_{\rm p}(x,x_{0},U_{\rm p})$ is independent of the column space of $\begin{bmatrix}\mathcal{H} & L\end{bmatrix}$ for all $x$. This is shown under two different assumptions.

\textbf{Sector-bounded nonlinearity: }Since $\meas{\phi}(x)$ is a sector bounded function, $\meas{\phi}(x)-\meas{\phi}(x_{0}) = M(x,x_0)(x-x_0)$ where
$M(x,x_0) = B + R(x,x_0)$. From Lemma~\ref{lem:MPsi}, this gives $\Phi_{\rm p}(x,U_{\rm p}) - \Phi_{\rm p}(x_{0},U_{\rm p})=\Psi_{\rm p}(x,x_{0},U_{\rm p}) (x-x_{0})$ with
\[
\resizebox{\columnwidth}{!}{$\begin{aligned}
\Psi_{\rm p}(x,x_{0},U_{\rm p}) 
&= \underbrace{\begin{bmatrix}
    B \\
    B \\
    \vdots\\
    B
    \end{bmatrix}}_{B_{\text{full}}}  + 
    \underbrace{\begin{bmatrix}
    R(x,x_0) \\
    R(x - f( u_{0}),x_{0} - f( u_{0})) \\
    \vdots\\
    R\left(x - \sum_{i=-p+2}^{0}f(u_{i}),x_{0} - \sum_{i=-p+2}^{0}f(u_{i})\right) 
    \end{bmatrix}}_{E(x,x_0,U_{\rm p})}.
\end{aligned}$}
\]
Since $B$ has full column rank, $B_{\text{full}}$ will have full column rank. Also, since  $\| R(x,x_{0}) \| \leq \gamma$, we will have $\|E(x,x_0,U_{\rm p})\| \le \gamma \sqrt{p}$.

Define
 \[
 f(E) =  \inf_{\theta,\|\psi\|=1}\| B_{\text{full}} \psi + E\psi + (\mathcal{H}K+L)\theta\|.
 \]
Note that $B_{\text{full}} \psi$ is a nonzero vector that repeats every $q$ entries. So it can be written as
\[
B_{\text{full}} \psi = \underbrace{\begin{bmatrix}
H(p)_{1,0} \\ H(p-1)_{1,0} \\ \vdots \\ H(1)_{1,0}
\end{bmatrix}}_{\mathcal{H}_B} k_B
\]
for some nonzero vector $k_B$ (in fact, $k_B=B_{\text{full}} \psi$). The matrix $[\mathcal{H} ~ \mathcal{H}_B]$ has $qp$ rows and $1 + \sum_{i=1}^{m}q(\ell_{i}+1)$ columns. From Lemma~\ref{lem:impulserepresentation}, $\sum_{i=1}^{m}(\ell_{i}+1) \le n$. Since $p > n$, it follows that $qp \ge 1 + \sum_{i=1}^{m}q(\ell_{i}+1)$.
By Theorem \ref{thm:hindependent} and the fact that the system has no eigenvalue at $\lambda=1$, it follows that $[\mathcal{H} ~ \mathcal{H}_B]$ has full column rank. So $B_{\text{full}} \psi$ is linearly independent of the columns of $\mathcal{H}$ for any non-zero $\psi$.

We now argue that the columns of $L$ are linearly independent of both the columns of $\mathcal{H}$ and the columns of $B_{\text{full}}$. Suppose that some column of $L$ can be expressed as a nonzero linear combination of the columns of $[\mathcal{H} ~ \mathcal{H}_B]$. Letting $Z$ denote the number of zeros at the top of this column of $L$, we have that the restriction of the aforementioned linear combination of the columns of $[\mathcal{H} ~ \mathcal{H}_B]$ must be $0$ on its first $Z$ entries. It follows that the columns of the restriction of $[\mathcal{H} ~ \mathcal{H}_B]$ to its first $Z$ rows must be linearly dependent. So if the restriction of $[\mathcal{H} ~ \mathcal{H}_B]$ to its first $Z$ rows has linearly independent columns, this prevents that column of $L$ from being in the column span of $[\mathcal{H} ~ \mathcal{H}_B]$. This linear independence is ensured by Theorem \ref{thm:hindependent} as long as $Z \ge 1 + \sum_{i=1}^m q(\ell_i + 1)$. By the construction of $L$, $Z \ge qp - q\ell_0$. So if $qp - q\ell_0 \ge 1 + \sum_{i=1}^m q(\ell_i + 1)$, the desired conclusion holds. Equivalently, we require $qp - q\ell_0 > \sum_{i=1}^m q(\ell_i + 1)$, which is equivalent to $p - \ell_0 > \sum_{i=1}^m (\ell_i + 1)$. From Lemma~\ref{lem:impulserepresentation}, we have that $\ell_{0} = n-\sum_{i=1}^{m}(\eta_{i}+1) \le n-\sum_{i=1}^{m}(\ell_{i}+1) < p - \sum_{i=1}^{m}(\ell_{i}+1)$. So $p-\ell_0 > \sum_{i=1}^{m}(\ell_{i}+1)$. 

Since the columns of $L$ are linearly independent of both the columns of $\mathcal{H}$ and the columns of $B_{\text{full}}$, when $E=0$, no $\theta$ exists such that the minimum of $f(E)$ is zero. That is, $f(0)>0$. We note that $f$ is a continuous function. Thus, there exists $\epsilon>0$ such that $f(E)>f(0)/2$ for all $\|E\|\leq \epsilon$.

Now,
 \begin{align*}
     \alpha &=  \inf_{x,\theta,\|\psi\|=1}\|\Psi_{\rm p}(x,x_{0},U_{\rm p})\psi-(\mathcal{H}K+L)\theta\|,\\
     & = \inf_{x,\theta,\|\psi\|=1}\left\|B_{\text{full}}\psi +E(x,x_0,U_{\rm p})\psi - (\mathcal{H}K+L)\theta\right\|,\\
     & =  \inf_{x} f(E(x,x_0,U_{\rm p})).
 \end{align*}
We have $\alpha>0$ when 
\[
\left\|E(x,x_{0},U_{\rm p})\right\|\leq \gamma \sqrt{p} \leq \epsilon.
\]
Thus the results holds for  $\gamma_{0} =\frac{\epsilon}{\sqrt{p}}$. 

\textbf{$L$-smooth nonlinearity:} From the first-order Taylor series expansion of $\meas{\phi}$ in~\eqref{eq:taylorJr}, it follows that
\[
\meas{\phi}(x)-\meas{\phi}(x_0) = M(x,x_0)(x-x_0)
\]
for
\begin{equation}
M(x,x_0) = J(x_0) + \frac{r(x,x_0)(x-x_0)^T}{\| x-x_0\|^2}.
\label{eq:MfullJ}
\end{equation}
From Lemma~\ref{lem:MPsi}, this then gives $\Phi_{\rm p}(x,U_{\rm p}) - \Phi_{\rm p}(x_{0},U_{\rm p})=\Psi_{\rm p}(x,x_{0},U_{\rm p}) (x-x_{0})$ with
\[
\resizebox{\columnwidth}{!}{$
\begin{aligned}
\Psi_{\rm p}(x,x_{0},U_{\rm p}) 
&= \underbrace{\begin{bmatrix}
    J(x_0) \\
    J(x_{0} - f( u_{0})) \\
    \vdots\\
    J\left(x_{0} - \sum_{i=-p+2}^{0}f(u_{i})\right)
    \end{bmatrix}}_{B(x_{0},U_{\rm p})} + \\
    & \quad
    \underbrace{\begin{bmatrix}
    \frac{r(x,x_0)(x-x_0)^T}{\| x-x_0\|^2} \\
    \frac{r(x - f( u_{0}),x_{0} - f( u_{0}))(x-x_0)^T}{\| x-x_0\|^2} \\
    \vdots \\
    \frac{r\left(x - \sum_{i=-p-2}^{0}f(u_{i}),x_{0} - \sum_{i=-p+2}^{0}f(u_{i})\right)(x-x_0)^T}{\| x-x_0\|^2} 
    \end{bmatrix}}_{E(x,x_0,U_{\rm p})}
\end{aligned}$}
\]
Note that $E(x,x_0,U_{\rm p})$ is a rank-one matrix:
\[
\resizebox{\columnwidth}{!}{$E(x,x_0,U_{\rm p}) = \begin{bmatrix}
    \frac{r(x,x_0)}{\| x-x_0\|} \\
    \frac{r(x - f( u_{0}),x_{0} - f( u_{0}))}{\| (x - f( u_{0}))-(x_{0} - f( u_{0}))\|} \\
    \vdots \\
    \frac{r\left(x - \sum_{i=-p+2}^{0}f(u_{i}),x_{0} - \sum_{i=-p+2}^{0}f(u_{i})\right)}{\| \left(x - \sum_{i=-p+2}^{0}f(u_{i})\right)-\left(x_{0} - \sum_{i=-p+2}^{0}f(u_{i})\right)\|} 
    \end{bmatrix} \cdot \frac{(x-x_0)^T}{\| x-x_0\|}$}
\]
This allows us to bound its spectral norm:
\[
\resizebox{\columnwidth}{!}{$
\begin{aligned}
\| E(x,x_0,U_{\rm p}) \| &= \left\| \begin{bmatrix}
    \frac{r(x,x_0)}{\| x-x_0\|} \\
    \frac{r(x - f( u_{0}),x_{0} - f( u_{0}))}{\| (x - f( u_{0}))-(x_{0} - f( u_{0}))\|} \\
    \vdots \\
    \frac{r\left(x - \sum_{i=-p+2}^{0}f(u_{i}),x_{0} - \sum_{i=-p+2}^{0}f(u_{i})\right)}{\| \left(x - \sum_{i=-p+2}^{0}f(u_{i})\right)-\left(x_{0} - \sum_{i=-p+2}^{0}f(u_{i})\right)\|} 
    \end{bmatrix} \right\| \underbrace{ \left\| \frac{(x-x_0)^T}{\| x-x_0\|} \right\|}_{=1} \\
    &\le \frac{L}{2} \sqrt{p} \| x-x_0 \|,
\end{aligned}$}
\]
where the last line uses Lemma~\ref{lem:Lsmooth}. 
Define
 \[
 f(E) =  \inf_{\theta,\|\psi\|=1}\| B(x_{0},U_{\rm p}) \psi + E\psi + (\mathcal{H}K+L)\theta\|.
 \]
If $J(x_k)$ is band-limited, there exists $k_{B}$ such that 
\[
B(x_{0},U_{\rm p}) =\mathcal{H}_{B}k_{B}
\]
where the columns of $\mathcal{H}_B$ are of the form 
\[
\begin{bmatrix}
    H(p)_{\lambda,\ell}\\
     H(p-1)_{\lambda,\ell}\\
     \vdots\\
      H(1)_{\lambda,\ell}
\end{bmatrix},
\]
with $|\lambda-1|<\rho$. (We use Lemma \ref{lem:remap} to shift the indices to 1 through $p$.) By Theorem \ref{thm:hindependent}, the matrix $\begin{bmatrix}\mathcal{H}_{B} & \mathcal{H}\end{bmatrix}$ is full column rank. Using the same argument as above, the columns of $L$ are linearly independent of both the columns of $\mathcal{H}$ and the columns of $\mathcal{H}_{B}$. Thus, when $E=0$, no $\theta$ exists such that the minimum of $f(E)$ is zero, and there exists $\epsilon$ such that $f(E)>0$ for $\|E\|<\epsilon$. Since $\alpha = \inf_{x}f(E(x,x_{0},U_{\rm p}))$, $\alpha>0$ for $ \| x-x_0 \|<\epsilon_{0}=\frac{2\epsilon}{L\sqrt{p}}$.

\end{proof}

\section{Conclusion}

This paper has presented bounds on the prediction error when utilizing the IREM model structure, which combines a linear convolution model with a nonlinear part that models integrator dynamics. The parameters of the IREM model are an equivalent input and the state of the integrator dynamics, and the prediction error of all outputs in the future is bounded in terms of how well these parameters fit a subset of the system outputs in the past. Conditions for the existence of these bounds provide observability conditions for the IREM model structure, while the gain of the bound is a quantitative measure of observability. While noise and disturbances were not explicitly considered in this paper, these bounds can be the basis for analysis of these effects once models for the noise and disturbances are chosen. 

\section{Acknowledgments}
    This work has benefited from collaborations with Prof. Robert Kee and Dr. Huayang Zhu of the Colorado School of Mines, and Dr. Tyler Evans of Astroscale U.S. 
    
\bibliographystyle{unsrt}
\bibliography{myBib}

 \begin{IEEEbiography}[{\includegraphics[width=1in,height=1.25in,clip,keepaspectratio]{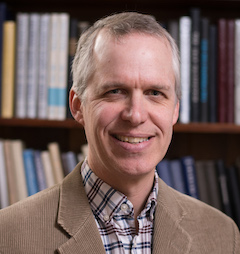}}]{Tyrone L. Vincent} (Senior Member, IEEE) received the B.S. degree in electrical engineering from the University of Arizona, Tucson, in 1992, and the M.S. and Ph.D. degrees in electrical engineering from the University of Michigan, Ann Arbor, in 1994 and 1997, respectively.
He is currently a Professor in the Department of Electrical Engineering at the Colorado School of Mines, Golden. His research interests include system identification, estimation, and control with applications in energy storage and generation.
\end{IEEEbiography}

 \begin{IEEEbiography}[{\includegraphics[width=1in,height=1.25in,clip,keepaspectratio]{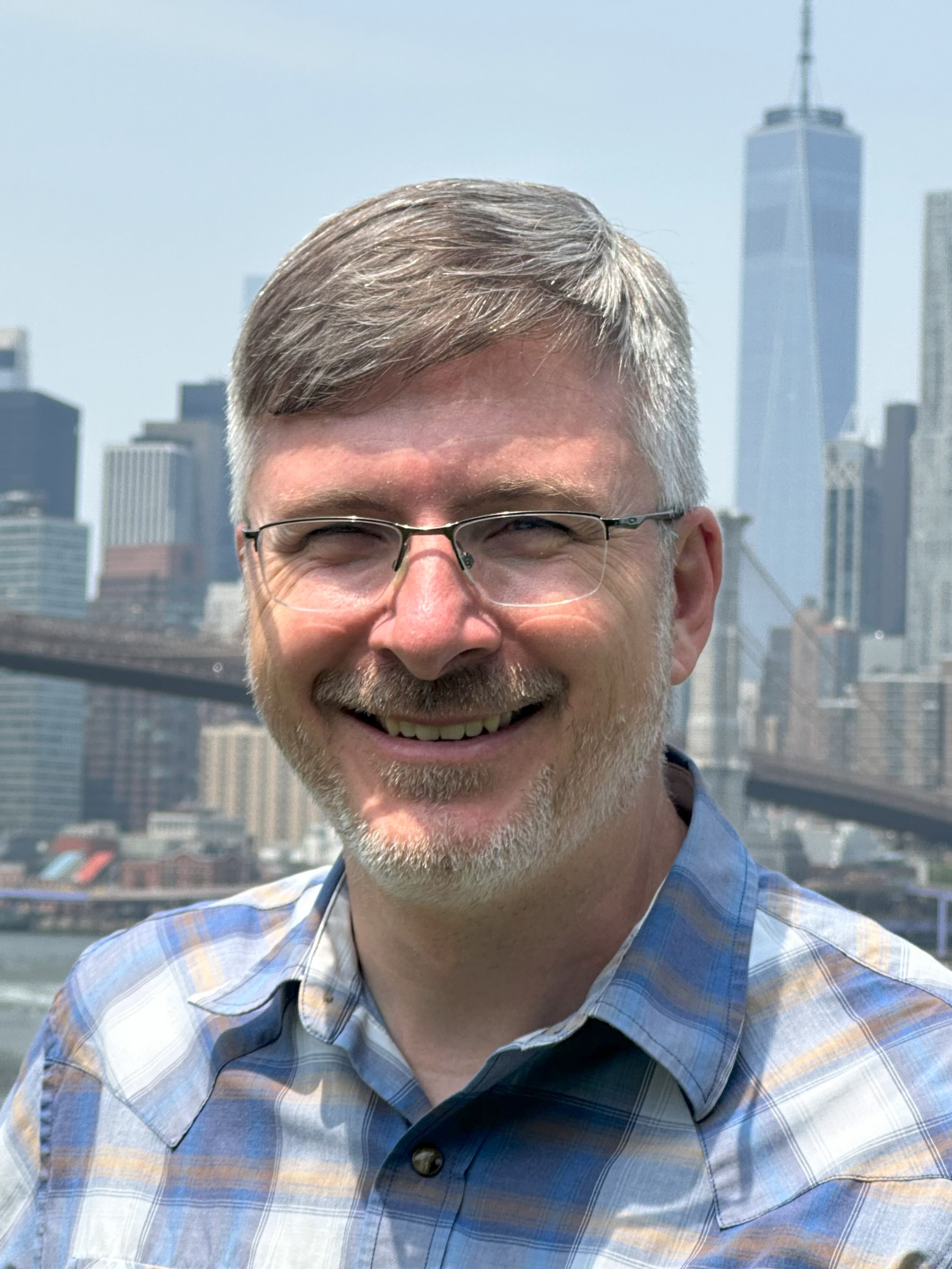}}]{Michael B. Wakin} (Fellow, IEEE) is a Professor of Electrical Engineering at the Colorado School of Mines. Dr. Wakin received a B.S. in electrical engineering and a B.A. in mathematics in 2000 (summa cum laude), an M.S. in electrical engineering in 2002, and a Ph.D. in electrical engineering in 2007, all from Rice University. He was an NSF Mathematical Sciences Postdoctoral Research Fellow at Caltech from 2006-2007, an Assistant Professor at the University of Michigan from 2007-2008, and a Ben L. Fryrear Associate Professor at Mines from 2015-2017. His research interests include signal/information processing and machine learning using sparse, low-rank, tensor, and manifold-based models.

\end{IEEEbiography}

\end{document}